\documentclass[11pt,twoside]{article}

\usepackage{a4}

\usepackage{amssymb,amsmath,amsthm,latexsym}
\usepackage{amsfonts}
  \usepackage{amsfonts}
\usepackage{graphicx}

\usepackage{mathtools,amsthm}
\newtheoremstyle{case}{}{}{}{}{}{:}{ }{}

\usepackage{hyperref}
\usepackage{amsmath, amsfonts}
\usepackage{amssymb, graphicx}
\usepackage{amscd}
\usepackage{textcomp}
\usepackage{palatino}
\usepackage{mathtools}
\usepackage{enumitem}
\usepackage[dvipsnames]{xcolor}
\usepackage{makecell}
\usepackage[pagewise]{lineno}%\linenumbers
\newtheorem{theorem}{Theorem}[section]

\newtheorem{definition}[theorem]{Definition}
\newtheorem{example}[theorem]{Example}
\newtheorem{lemma} [theorem]{Lemma}

\newtheorem{remark}[theorem]{Remark}

\numberwithin{subcase}{case}

\vspace{5cm}

\begin{document}
  
  \label{'ubf'}  
\setcounter{page}{1}                                 %Put here the starting page number

\markboth {\hspace*{-9mm} \centerline{\footnotesize \sc
          $g$-Circulant Matrices
         }
                 }
                { \centerline                           {\footnotesize \sc                     
      T. Chatterjee and A. Laha } \hspace*{-9mm}              
               }

\vspace*{-2cm}
%\begin{flushleft}
  %  {\footnotesize\it J. of the Ramanujan Mathematical Society,  XX, No. XX (XXXX),  pp. XX--XX \\ %\pageref{'ubf'}-\pageref{'ubl'}\\  
    % Use the appropriate labels
 %   }
    %       \vspace*{0.3cm}       %{2.2cm}
%\end{flushleft}

\begin{center}
{ 
       { \textbf {On MDS Property of g-Circulant Matrices
    % Put the title of the paper here
                               }
       }
\\

\medskip
{\sc Tapas Chatterjee }\\
{\footnotesize Indian Institute of Technology Ropar, Punjab, India.
}\\
{\footnotesize e-mail: {\it tapasc@iitrpr.ac.in}}
\medskip

{\sc Ayantika Laha }\\
{\footnotesize Indian Institute of Science, Bangalore, India.
}\\
{\footnotesize e-mail: {\it 2018maz0008@iitrpr.ac.in\\ \it ayantika.laha.math@gmail.com}}
\medskip
}
\end{center}

\thispagestyle{empty} 
\vspace{-.4cm}

\hrulefill

\begin{abstract}  
Circulant Maximum Distance Separable (MDS) matrices have gained significant importance due to their applications in the diffusion layer of the AES block cipher. In $2013$, Gupta and Ray established that circulant involutory matrices of order greater than $3$ cannot be MDS over $\mathbb{F}_{2^m}$. This finding prompted a generalization of circulant matrices and the involutory property of matrices by various authors. In $2016$, Liu and Sim introduced cyclic matrices by changing the permutation of circulant matrices. In $1961,$ Friedman introduced $g$-circulant matrices which form a subclass of cyclic matrices. 
In this article, we first discuss $g$-circulant matrices with involutory and MDS properties. We prove that $g$-circulant involutory matrices of order $k \times k$ cannot be MDS unless $g \equiv -1 \pmod k.$
Next, we delve into $g$-circulant semi-involutory and semi-orthogonal matrices with entries from finite fields. We establish that the $k$-th power of the associated diagonal matrices of a $g$-circulant semi-orthogonal (semi-involutory) matrix of order $k \times k$ results in a scalar matrix. These findings extend the recent results on circulant matrices established by Kumar {\it{et al.}} $(2026)$ and Chatterjee {\it{et al.}} $(2022)$. Furthermore, we prove that cyclic matrices of order $2^{d} \times 2^{d}$ over finite fields of characteristic $2$ cannot simultaneously possess both the MDS and semi-orthogonal properties.
\end{abstract}
\hrulefill



\vspace{-.37cm}
\section{Introduction}

Symmetric key primitives like block ciphers, stream ciphers, and hash functions rely on various components that provide confusion and diffusion. These two components are important for the overall security and efficiency of the cryptographic scheme. The goal of the confusion layer is to obscure the relationship between the plaintext and the ciphertext, while the diffusion layer spreads the influence of each input symbol in such a way that every output symbol depends on many input symbols. Typically, the diffusion layer of a cipher is implemented using a diffusion matrix, and its strength is evaluated by how effectively it spreads the influence of each input symbol across the output. As a result, constructing diffusion layers that achieve strong diffusion while allowing low-cost implementations remains a key challenge in cipher design. Many block ciphers like AES \cite{DR}, LED \cite{GPPR}, SHARK \cite{RDPB}, SQUARE \cite{DKR} and many hash functions like PHOTON \cite{GPP}, WHIRLPOOL \cite{BR} use maximum distance separable (MDS) matrices in their diffusion layers to achieve optimal diffusion.

%There are two main approaches to constructing an MDS matrix. The first method is a recursive construction, where a companion matrix or a sparse matrix $A$ of order $n \times n$ is used, such that $A^n$ is an MDS matrix. This method has been employed by the block cipher LED \cite{GPPR} and the hash function PHOTON \cite{GPP}. In non-recursive constructions, the constructed matrices are inherently MDS. In such cases, two main techniques prevail. The first involves a search method by enumerating a specific search space. Alternatively, various matrix structures, including Cauchy matrices, Vandermonde matrices, Hadamard matrices, etc., are used. Another method of constructing MDS matrices involves circulant matrices, which was first proposed by Daemen {\it{et al.}} in $1998$ for the diffusion layer of the AES block cipher. Subsequently, in $2003$, Barreto {\it{et al.}} applied the same idea to the Whirlpool hash function. Since then, the search for efficient circulant MDS matrices has garnered considerable attention.

There are several approaches to constructing MDS matrices. One important method is the recursive construction, where a companion matrix or a sparse matrix $A$ of order $n \times n$ is used such that $A^n$ is an MDS matrix. This method has been employed in the block cipher LED \cite{GPPR} and the hash function PHOTON \cite{GPP}. 
Another approach consists of direct (non-recursive) constructions, where the objective is to construct matrices that are inherently MDS under certain algebraic conditions, without relying on repeated powers or iterative procedures. In this setting, search-based techniques are commonly used by enumerating matrices within a prescribed search space. At the same time, structured algebraic constructions based on Cauchy matrices, Vandermonde matrices, Hadamard matrices, etc. have been extensively studied in the literature \cite{GR2,GR,GPRS}. %In addition, several hybrid approaches combining algebraic structures with computational search techniques have been proposed in the literature.
Among these structured constructions, circulant matrices form an important and well-studied class. The use of circulant matrices for constructing MDS matrices was first proposed by Daemen {\it et al.} in $1998$ for the diffusion layer of the AES block cipher \cite{DR}. Subsequently, in $2003$, Barreto {\it et al.} \cite{BR} applied the same idea to the Whirlpool hash function. Since then, the construction of efficient circulant MDS matrices has attracted considerable attention.

Note that the inverse of a diffusion matrix plays a pivotal role in the decryption process of a substitution-permutation network (SPN) based block cipher. Therefore, an essential property for an MDS matrix is that its inverse can be efficiently implemented, meaning that it can be computed using a minimal number of arithmetic operations (additions and multiplications) and with low memory usage. Consequently, MDS matrices with involutory or orthogonal properties have become particularly important, as these structures allow the inverse to be computed directly or with minimal additional computation, enhancing both the implementability and efficiency of cryptographic applications.
In \cite{GR1, GR2, GR}, Gupta {\it{et al.}} studied circulant MDS matrices with involutory and orthogonal properties over the finite field $\mathbb{F}_{2^m}.$ They proved that circulant orthogonal matrices of order $2^d \times 2^d$ cannot be MDS, and there does not exist any circulant involutory matrix of order $n \geq 3$ which is MDS. In the search for circulant involutory MDS matrices, Gupta {\it et al.} \cite{GR} introduced two types of matrices containing circulant submatrices. Type-I matrices are neither involutory nor orthogonal \cite{GR, GPRS}, whereas Type-II matrices are involutory over the finite field $\mathbb{F}_{2^m}$. Later, in $2016,$ Liu and Sim \cite{LS} provided examples of involutory MDS left-circulant matrices of various orders over the finite fields $GF(2^8)$ and $GF(2^4).$ They also defined cyclic matrices which are a generalization of circulant matrices by changing the permutation. In \cite{TA2}, the authors proved a representation of cyclic matrices using permutation matrices which can be seen as a generalization of the representation of $g$-circulant matrices \cite{D}. Introduced by B. Friedman in $1961$ \cite{F}, $g$-circulant matrices form a subclass of cyclic matrices and also reduced to left-circulant matrices when $g \equiv -1 \pmod{n}$, where $n$ is the matrix order. The study of the MDS property of cyclic matrices was motivated by \cite{LS}, which demonstrated the construction of left-circulant MDS matrices of certain orders. This naturally raises the question of whether similar constructions exist for general cyclic or $g$-circulant matrices. In this context, \cite{TA2} investigated the orthogonal property of $g$-circulant MDS matrices using their general permutation-based representation.

Fiedler {\it{et al.}} introduced the notion of semi-orthogonal matrices in \cite{FH} and Cheon {\it{et al.}} \cite{CCK} introduced the notion of semi-involutory matrices. It is worth noticing that these matrices have simple computable inverses. In \cite{TA}, Chatterjee {\it{et al.}} proved that the Cauchy-based construction of an MDS matrix is semi-orthogonal over a finite field. They also studied circulant matrices in the context of semi-orthogonal and semi-involutory properties. 

Given the extensive research on circulant matrices and the numerous non-existence results regarding MDS, involutory, and orthogonal properties, investigating circulant MDS matrices with semi-involutory or semi-orthogonal properties holds significant promise. The existence of such matrices could lead to efficient implementations of the inverse of circulant matrices, provided that the diagonal matrices corresponding to semi-involutory and semi-orthogonal matrices are efficient, with entries incurring low implementation costs. In \cite{TA3, KSMG}, a comprehensive study of circulant semi-involutory and semi-orthogonal MDS matrices over finite fields of characteristic $2$ was carried out. Meanwhile, in \cite{TA1}, the authors characterized $3 \times 3$ semi-involutory MDS matrices over finite fields of characteristic $2$, while in \cite{KSMG}, a similar characterization was obtained for $3 \times 3$ semi-orthogonal MDS matrices.         

Many authors have continued the search for MDS matrices from finite fields to rings and modules. In $1995$, Zain and Rajan defined MDS codes over cyclic groups \cite{ZR} and Dong {\it{et al.}} characterized MDS codes over elementary Abelian groups \cite{DCG}. By considering a finite Abelian group as a torsion module over a PID, Chatterjee {\it{et al.}} proved some non-existent results of MDS matrices in $2022$ \cite{TAS}. In \cite{CL}, Cauchois {\it{et al.}} introduced $\theta$-circulant matrices over the quasi-polynomial ring. They proposed a construction for $\theta$-circulant almost-involutory MDS matrix over the quasi-polynomial ring.

\section{Contribution}
In \S~\ref{sec:g-circulant MDS involutory} of this article, we begin by establishing the structure of the matrix $A^2$ using permutation matrices, where $A$ is a $g$-circulant matrix. Subsequently, we explain the rationale behind focusing on the condition $g^2 \equiv 1 \pmod k$ to construct an involutory MDS $g$-circulant matrix of order $k \times k.$ We then proceed to prove that, among all solutions of the equation $g^2 \equiv 1 \pmod k$, the feasibility of constructing an involutory MDS matrix of order $k \times k$ is limited to case $g \equiv -1 \pmod k.$
 Additionally, in \S~\ref{sec:g-circulant MDS semi-involutory}, we explore $g$-circulant semi-orthogonal matrices with entries from finite fields. Furthermore, we investigate $g$-circulant semi-involutory matrices. We also prove that $g$-circulant semi-orthogonal matrices of order $2^{d} \times 2^{d}$ over finite fields of characteristic $2$ cannot be MDS. The same conclusion holds for cyclic semi-orthogonal matrices of the same order.
 %In both scenarios, we establish that the $k$-th power of the associated diagonal matrices becomes a scalar matrix when the order of the matrix is $k \times k.$

\section{Preliminaries} 
 In this section, we describe the notations and important definitions we use throughout the paper. 
 
We begin with some notations and definitions from \cite{MS}.
Let $\mathbb{F}_q$ denote a finite field with $q$ elements where $q$ is power of a prime $p$ and $\mathbb{F}^*$ denote the non-zero elements of a finite field. Let $\mathcal{C}$ be a $[n, k, d]$ linear error correcting code over the finite field $\mathbb{F}_q$ with length $n$, dimension $k$, and minimum Hamming distance $d.$ The code $\mathcal{C}$ is a $k$ dimensional subspace of $\mathbb{F}_q^n.$ The generator matrix $G$ of $\mathcal{C}$ is a $k \times n$ matrix with the standard form $[I|A]$, where $I$ is a $k \times k$ identity matrix and $A$ is a $k \times n-k$ matrix.  The Singleton bound states that, for an $[n,k,d]$ code, $n-k \geq d-1.$ An $[n,k,n-k+1]$ code is called a maximum distance separable (MDS) code. Another definition of an MDS code in terms of generator matrix, as given in \cite{MS}, is the following. 
\begin{theorem}
An $[n,k,d]$ code $\mathcal{C}$ with the generator matrix $G=[I|A]$, where $A$ is a $k \times (n-k)$ matrix, is MDS if and only if every $i \times i$ submatrix of $A$ is non-singular, $i=1,2,\hdots,\text{min}(k,n-k).$
\end{theorem}
 This definition of MDS code gives the following characterization of an MDS matrix.

\begin{definition}
A square matrix $A$ is said to be MDS if every square submatrix of $A$ is non-singular.
\end{definition}

MDS matrices with efficiently implementable inverses are useful because of the use of inverse matrix in the decryption layer of an SPN based block cipher. One approach to achieve efficient implementation is to use MDS matrices with either involutory or orthogonal property. Here, $A^{-1}$ denote the inverse of $A$, $A^{T}$ denote the transpose of $A$ and $I$ is the identity matrix.

\begin{definition}
A square matrix $A$ is said to be involutory if $A^2=I$ and orthogonal if $AA^T=A^TA=I.$
\end{definition}

In $2012,$ Fielder {\it{et al.}} generalized the orthogonal property of matrices to semi-orthogonal in \cite{FH}. The definition of a semi-orthogonal matrix is as follows.

\begin{definition}
A non-singular matrix $M$ is semi-orthogonal if there exist non-singular diagonal matrices $D_1$ and $D_2$ such that
$M^{-T}=D_1MD_2$, where $M^{-T}$ denotes the transpose of the matrix $M^{-1}.$
\end{definition}
The following property concerning semi-orthogonal matrices and permutation matrices was proved by Fiedler {\it{et al.}} in \cite{FH}. Note that, a permutation matrix is a square matrix that is obtained by permuting the rows (columns) of the identity matrix.  Moreover, permutation matrices are orthogonal.

\begin{lemma}\label{per equi s.o}
If $A$ is semi-orthogonal and $P$ is a permutation matrix, then both $PA$ and $AP$ are semi-orthogonal.
\end{lemma}

Following that, in $2021,$ Cheon {\it{et al.}} \cite{CCK}  defined semi-involutory matrices as a generalization of the involutory matrices. The definition of a semi-involutory matrix is as follows.
\begin{definition}
A non-singular matrix $M$ is said to be semi-involutory if there exist non-singular diagonal matrices $D_1$ and $D_2$ such that $M^{-1} = D_1MD_2.$  
\end{definition}
An analogous result to Lemma \ref{per equi s.o} has been established by Cheon {\it{et al.}}, as follows.
\begin{lemma}
$A$ is semi-involutory if and only if $P^TAP$ is semi-involutory for any permutation matrix $P.$
\end{lemma}

Similar to the semi-orthogonal and semi-involutory properties, the MDS property is invariant under permutation equivalence. The following result, established in \cite{GPRS}, formalizes this observation.
\begin{lemma}\label{mds PE}
If $A$ is an MDS matrix, then for any permutation matrices $P$ and $Q$,
$PAQ$ is an MDS matrix.
\end{lemma}

As previously mentioned, circulant matrices find application in the diffusion layer. In this context, we now provide definitions for circulant matrices and their generalizations. Let $A$ be an $n\times n$ matrix. The $i$-th row of $A$ is denoted by $R_i$ for $0 \leq i \leq n-1$ and the $j$-th column as $C_j$ for $0 \leq j \leq n-1.$ Furthermore, $A[i,j]$ denotes the entry at the intersection of the $i$-th row and $j$-th column. The definition of a circulant matrix is the following. 

\begin{definition}
The square matrix of the form $\begin{bmatrix}
c_0 & c_1 & c_2 & \cdots & c_{k-1}\\
c_{k-1} & c_0 & c_1 & \cdots & c_{k-2}\\
\vdots & \vdots & \vdots & \cdots & \vdots\\
c_1 & c_2 & c_3 & \cdots & c_0
\end{bmatrix}$ is said to be circulant matrix and denoted by $\mathcal{C}=$ circulant$(c_0 , c_1 , c_2 , \hdots , c_{k-1})$ On the other hand, the square matrix of the form $\begin{bmatrix}
c_0 & c_1 & c_2 & \cdots & c_{k-1}\\
c_{1} & c_2 & c_3 & \cdots & c_{0}\\
\vdots & \vdots & \vdots & \cdots & \vdots\\
c_{k-1} & c_0 & c_1 & \cdots & c_{k-2}
\end{bmatrix}$ is said to be left-circulant matrix and denoted by left-circulant$(c_0 , c_1 , c_2 , \hdots , c_{k-1}).$
\end{definition}

 The entries of the circulant matrix $\mathcal{C}$can be expressed as $\mathcal{C}[i,j]=c_{j-i+1},$ where subscripts are calculated modulo $k.$
The representation of circulant matrices using permutation matrices is as follows: 
\begin{eqnarray}\label{circulant structure}
\mathcal{C}= \text{circulant} (c_0 , c_1 , c_2 , \hdots , c_{k-1})=c_0I+c_1P+c_2P^2+\cdots +c_{k-1}P^{k-1},
\end{eqnarray}
where $I$ denotes the $k \times k$ identity matrix and $P=$ circulant$(0,1,0,\hdots,0)$, a permutation matrix of order $k \times k.$

In \cite{TA}, Chatterjee {\it{et al.}} proved the following two properties of the associated diagonal matrices of circulant semi-orthogonal and semi-involutory matrices with entries from non-zero elements over any finite field. Note that, the additional observation that all entries of the matrix $A$ must be non-zero was made in Remark $5$ of \cite{KSMG}.
\begin{theorem}\label{circ semi-inv}
Let $A$ be an $k \times k$ circulant matrix over a finite field. Then $A$ is semi-involutory if and only if there exist non-singular diagonal matrices $D_1,D_2$ such that $D_1^k=k_1I$ and $D_2^k=k_2I$ for non-zero scalars  $k_1,k_2$ in the finite field, and $A^{-1}=D_1AD_2.$ 
\end{theorem}

\begin{theorem}\label{circ semi-ortho}
Let $A$ be a $k \times k$ circulant matrix over a finite field $\mathbb{F}$. Then $A$ is semi-orthogonal if and only if there exist non-singular diagonal matrices $D_1$ and $D_2$ such that $D_1^k=k_1I$ and $D_2^k=k_2I$ for non-zero scalars  $k_1,k_2 \in \mathbb{F}$ and $A^{-T}=D_1AD_2.$
\end{theorem}
In $1961,$ Friedman introduced a generalization of the circulant matrix in \cite{F}, termed the $g$-circulant matrix. In this matrix, each row (except the first) is derived from the previous row by cyclically shifting the elements by $g$ columns to the right. The formal definition is provided below.

 \begin{definition}\label{g-circ def}
 A $g$-circulant matrix of order $k \times k$ is a matrix of the form $A=$ g-circulant$(c_0,c_1,\hdots,c_{k-1})=\begin{bmatrix}
c_0 & c_1 & \cdots & c_{k-1}\\
c_{k-g} & c_{k-g+1} & \cdots & c_{k-1-g}\\
c_{k-2g} & c_{k-2g+1} & \cdots & c_{k-1-2g}\\
\vdots& \vdots &\cdots&\vdots\\
c_{g} & c_{g+1} & \cdots & c_{g-1}\\
\end{bmatrix}$, where all subscripts are taken modulo $k.$ 
\end{definition}
Entries of a $g$-circulant matrix satisfy the relation $A[i,j]=A[i+1,j+g]$, where subscripts are calculated modulo $k.$  Moreover, for a $g$-circulant matrix $A=(a_{i,j}),~0 \leq i,j \leq k-1$ with first row $(c_0,c_1,\hdots,c_{k-1})$, we have $A[i,j]=c_{j-ig \pmod k}.$

For $g=1$, a $g$-circulant matrix represents a circulant matrix, and for $g\equiv -1 \pmod{k}$, it takes the form of a left-circulant matrix.
Some noteworthy properties of $g$-circulant matrices are provided in \cite{AB, D}.

\begin{lemma}\label{gh-circulant}
Let $A$ be $g$-circulant and $B$ $h$-circulant. Then $AB$ is $gh$-circulant.
\end{lemma}

\begin{lemma}\label{PA=AP^g}
$A$ is $g$-circulant if and only if $PA=AP^g$ where $P$ is the permutation matrix $P=$ circulant$(0,1,0,\hdots,0).$
\end{lemma}
The classification of $g$-circulant matrices of order $k \times k$ is divided into two types depending on $\gcd(k,g).$ In this article, we only consider the case $\gcd(k,g)=1$ because, if $\gcd(k,g) >1$, $g$-circulant matrices cannot be MDS as proved in Theorem $4.1$ of \cite{TA2}. For the scenario where $\gcd(g,k)=1$, the inverse of a non-singular $g$-circulant matrix exhibits a specific characteristic, as proven in \cite{D}.  This characteristic is noted in the following lemma:
\begin{lemma}\label{A-inv g-inv circ}
Let $A$ be a non-singular $g$-circulant matrix of order $k \times k$ with $\gcd(g,k)=1.$ Then $A^{-1}$ is $g^{-1}$-circulant, where $g^{-1}$ denotes the multiplicative inverse of
$g$ modulo $k$.
\end{lemma}
The transpose of a g-circulant matrix exhibits similar characteristics.
\begin{lemma}\label{transpose g circ}
Let $A$ be a $g$-circulant matrix of order $k\times k$ with $\gcd(g,k)=1.$ Then $A^{T}$ is $g^{-1}$-circulant, where $g^{-1}$ denotes the multiplicative inverse of
$g$ modulo $k$.
\end{lemma}

\begin{proof}
Given that $A=(a_{i,j}), 0 \leq i,j \leq k-1$ is $g$-circulant matrix, we have $A[i,j]=A[i+1,j+g]$, i.e., $a_{i,j}=a_{i+1,j+g}$ for all $0 \leq, i,j \leq k-1,$  considering subscripts modulo $k.$
 Using this property, the entries of $A$ exhibit the following pattern: $$a_{i,j}=A[i,j]=A[i+1,j+ g]=A[i+2,j+ 2g]=\cdots=A[i+l,j+ lg ]=a_{i+l,j+lg},$$  where subscripts are calculated modulo $k.$ Consequently, the entry at the $i$-th row and $j$-th column repeats at the $j+1$-th column %and the $l$-th row with $0 \leq l \leq k-1, l \neq i$ 
 when $j+lg=j+1 \pmod k.$ Since $\gcd (g,k)=1$, we have $l=g^{-1}.$ Then entries of $A^T$ are $A[j,i]$ and they satisfy $A[j,i]=A[j+1,i+g^{-1}]$ for all $0 \leq, j,i \leq k-1.$ Thus $A^{T}$ is $g^{-1}$-circulant.
\end{proof}
% Then $$A[i,0]=A[i+1, g]=A[i+2, 2g]=\cdots=A[i+l, lg ]$$, with the subcripts are calculated modulo $k.$ Then the entry at position $A[i,0]$ will repeat at the second column at the position $l$ when $lg=1 \pmod k$. Since $\gcd (g,k)=1$, then $l=g^{-1}$.
Moreover, in the case of $\gcd(g,k)=1$, the representation of $g$-circulant matrices using permutation matrices is elucidated in \cite{D} as follows:
\begin{eqnarray}\label{g-circulant structure}
A= g\text{-circulant}(c_0 , c_1 , c_2 , \hdots , c_{k-1})=c_0Q_g+c_1Q_gP+c_2Q_gP^2+\cdots +c_{k-1}Q_gP^{k-1},
\end{eqnarray}
where $Q_g=g$-circulant$(1,0,0,\hdots,0)$, $P=$ circulant$(0,1,0,\hdots,0)$, and both are permutation matrices of order $k \times k.$

The notion of the cyclic matrix was introduced by Liu and Sim \cite{LS} as a generalization of the circulant matrix in $2016.$ 
A cyclic matrix of order $k \times k$ is defined using a $k$-cycle permutation $\rho$ of its first row, where $\rho  \in S_k$, the symmetric group of $k$ elements. 
The definition of cyclic matrix is the following.

\begin{definition}\label{cyclic def}
For a $k$-cycle $\rho \in S_k$, a matrix $\mathfrak{C}_\rho$ of order $k \times k$ is called cyclic if each subsequent row is $\rho$-permutation of the previous row. We represent this matrix as $\text{cyclic}_\rho(c_0,c_1,c_2,\hdots,c_{k-1})$, where $(c_0,c_1,c_2,\hdots,c_{k-1})$ is the first row of the matrix. The $(i,j)$-th entry of $\mathfrak{C}_\rho$ can be expressed as $\mathfrak{C}_\rho(i,j)=c_{\rho^{-i}(j)}.$
\end{definition}

For example, the matrix $cyclic_\rho(c_0,c_1,c_2,\hdots,c_{k-1})$, where $\rho=(0 ~1 ~ 2 \cdots~ k-1) \in S_k$ results in a circulant matrix. Similarly, if we use $\rho =(0 ~ k-1 ~1 ~2 \cdots k-2) \in S_k$, we obtain a left-circulant matrix.  Note that, a $k$-cycle of the form $\begin{pmatrix} \label{g-cycle}
 0 & 1 & 2 & \cdots &k-1\\
 g & g+1 & g+2 & \cdots & g+k-1
 \end{pmatrix}$, where $g+i$ is calculated modulo $k$ and $\gcd(k,g)=1$ can be written as $(0 \quad g \quad {2g\pmod k} \quad {3g\pmod k} \cdots {(k-1)g \pmod k}).$ This gives a complete $k$- cycle because of the next lemma. 
 \begin{lemma}\label{gcd result}
Let $S= \{\alpha g \pmod k, ~\alpha =0,1,\hdots ,k-1\}.$ $S$ will be a complete residue system modulo $k$ if and only if $\gcd(k,g)=1.$
\end{lemma}
Cyclic matrices corresponding to these cycles are $g$-circulant matrices. In \cite{LS}, Liu and Sim proved a bijection between cyclic and circulant matrices, although no explicit structure of the corresponding permutation matrix was provided.
Later, in \cite{TA2}, Chatterjee and Laha established the permutation equivalence between cyclic and circulant matrices by explicitly determining the structure of the associated permutation matrix. The result is as follows:
\begin{theorem}\label{relation cyclic circulant}
Let $\mathfrak{C}_\rho(c_0,c_1,\hdots,c_{k-1})$ be a cyclic matrix. Then there exists a unique permutation matrix $Q$ defined by $$Q(i,j)=\begin{cases}
  1, & \text{if $i=\rho^j(0),j=0,1,\cdots, k-1$; }\\
  0, & \text{otherwise}.
\end{cases}$$ such that $\mathfrak{C}_\rho Q=$ circulant$(c_{0},c_{\rho(0)},c_{\rho^2(0)},c_{\rho^3(0)},\hdots,c_{\rho^{k-1}(0)}).$
Moreover, $Q^{-1}=\text{cyclic}_\rho (1,0,0,\hdots,0).$
\end{theorem}

Utilizing Theorem \ref{relation cyclic circulant}, they derived a representation for cyclic matrices using permutation matrices in \cite{TA2}:
\begin{eqnarray}\label{cyclic structure}
\mathfrak{C}= \text{cyclic}_{\rho} (c_0 , c_1 , c_2 , \hdots , c_{k-1})=c_0Q_{\rho}+c_1Q_{\rho}P+c_2P^2Q_{\rho}+\cdots +c_{k-1}P^{k-1}Q_{\rho},
\end{eqnarray}
where $Q_{\rho}=$ cyclic$(1,0,0,\hdots,0)$, $P=$ circulant$(0,1,0,\hdots,0)$, and both are permutation matrices of order $k \times k.$

\section{\bf g-Circulant matrices with MDS and involutory properties \label{sec:g-circulant MDS involutory}}

Liu and Sim \cite{LS} provided examples of involutory MDS left-circulant matrices for odd order over the finite field $\mathbb{F}_{2^m}.$ They showed that, there are no involutory MDS cyclic matrices of order $4, 8.$ In this section, we extend their investigation to a subclass of cyclic matrices and establish a structural characterization of involutory $g$-circulant matrices.
Our results provide a theoretical framework explaining why odd and even order $g$-circulant matrices exhibit different behaviours under the involutory and MDS conditions.

\begin{theorem}\label{A^2 g-circulant structure}
Let $A$ be a $g$-circulant matrix with the first row $(c_0,c_1,\hdots,c_{k-1})$ and $\gcd(k,g)=1.$ Then $A^2$ can be expressed as $$A^2=\sum_{l=0}^{k-1}\left( \sum\limits_{\substack{i,j=0\\gi+j=l\pmod k}} ^{k-1}c_ic_j \right) Q_g^2P^l,$$ where $Q_g^2=g^2$-circulant$(1,0,0,\hdots,0)$ and $P=$circulant$(0,1,0,\hdots,0).$
\end{theorem}
 
\begin{proof}
Let $A= g$-circulant$(c_0,c_1,\hdots,c_{k-1})$ with $\gcd(k,g)=1.$ Then by Equation \ref{g-circulant structure}, $A$ can be expressed as
$A=\sum_{i=0}^{k-1} c_iQ_gP^{i},$
where $Q_g$ is a $g$-circulant matrix.
% Since $Q_g^2$ is $g^2 (\mod k)$-circulant matrix
Therefore $A^2$ can be written as
\begin{align*}
A^2&=(c_0Q_g+c_1Q_gP+c_2Q_gP^2+c_3Q_gP^3+\cdots+ c_{k-1}Q_gP^{k-1})^2\\
&=c_0^2Q_g^2+ c_1^2(Q_gP)^2+\cdots+c_{k-1}^2(Q_gP^{k-1})^2+c_0c_1Q_gQ_gP+c_0c_2Q_gQ_gP^2+\cdots\\& \quad +c_{k-2}c_{k-1}Q_gP^{k-2}Q_gP^{k-1}
\end{align*}
Using the identity $PQ_g=Q_gP^g$ and $P^k=I$, we can derive that $Q_gP^iQ_gP^{k-ig}=Q_g^2P^k=Q_g^2.$ Therefore, the coefficient of $Q_g^2$ in $A^2$ is: 
\begin{eqnarray*}
\sum_{\substack{i,j=0,\\gi+j=0 \pmod k}}^{k-1}c_ic_j=c_0^2+c_1c_{k-g}+c_2c_{k-2g}+\cdots+c_{k-1}c_{k-(k-1)g}.
\end{eqnarray*}
Similarly, the coefficient of $Q_g^2P$ in $A^2$ can be written as 
\begin{eqnarray*}
\sum_{\substack{i,j=0,\\gi+j=1 \pmod k}}^{k-1}c_ic_j=c_0c_1+c_1c_{1
+k-g}+c_2c_{1+k-2g}+\cdots+c_{k-1}c_{1+k-(k-1)g}.
\end{eqnarray*}
Thus using induction, we get the coefficient of $Q_g^2P^l$ is $\sum\limits_{\substack{i,j=0,\\ gi+j=l \pmod k}}^{k-1}c_ic_j.$ Therefore, we can conclude that $A^2=\sum\limits_{l=0}^{k-1}\left( \sum\limits_{\substack{i,j=0\\gi+j=l\pmod k}} ^{k-1}c_ic_j \right) Q_g^2P^l.$
\end{proof}
\noindent
This theorem provides an explicit description of the entries of the matrix $A^2$.
In particular, for the matrices $Q_g^2 P^l$, with $0 \le l \le k-1$, the first row
contains a single nonzero entry equal to $1$ in the $l$-th position. As a result,
the first row of $A^2$ admits a simple expression. Specifically, for
$0 \le r \le k-1$,
$
A^2[0,r]
=
\sum_{{i,j=0~ (gi+j) \equiv r \pmod{k}}}^{k-1}
c_i c_j .
$
Moreover, since $A^2$ is a $g^2$-circulant matrix, its entries satisfy
$
A^2[i,j] = A^2[i+1,\, j+g^2],
$
where all indices are taken modulo $k$. Consequently, once the first row of $A^2$
is determined, all remaining entries follow immediately from this relation.

Utilizing the structure of $A^2$ from Theorem \ref{A^2 g-circulant structure}, we discuss the existence of $g$-circulant matrices over the finite field $\mathbb{F}_{2^m}$ with both involutory and MDS properties. To begin with, we show the non-existence of $g$-circulant involutory matrices when $g^2 \not\equiv 1 \pmod {k}.$ The theorem is as follows.

\begin{theorem}\label{not inv g circ}
Let $A$ be a $g$-circulant matrix of order $k \times k$ and $\gcd(k,g)=1.$ If $g^2 \not\equiv 1 \pmod {k}$, then $A$ cannot be involutory.
\end{theorem}

\begin{proof}
Let $A$ be a $g$-circulant matrix of order $k \times k$ and $\gcd(k,g)=1.$ Then $A^2$ is a $g^2$-circulant matrix. Therefore $A^2[0,0]=A^2[1,g^2].$ If $A$ is involutory then $A^2[0,0]=1.$ But $A^2[1,g^2]=0$ since $g^2 \not\equiv 1 \pmod {k}.$ This is a contradiction. 
\end{proof}

An example illustrating Theorem \ref{not inv g circ} is as follows.
\begin{example}\label{ex1}
Let $a$ be a root of the primitive polynomial
$
1+x^2+x^5+x^6+x^8
$
over $\mathbb{F}_2$. Then $a$ is a primitive element of the finite field $\mathbb{F}_{2^8}$.
%Let $a$ be a primitive element of the finite field $\mathbb{F}_{2^8}$ with the generating polynomial $1+x^2+x^5+x^6+x^8.$ 
Consider the $3$-circulant matrix of order $5\times 5$ with the first row $(1,a,1+a+a^4+a^5+a^7,1+a+a^3+a^4+a^5+a^7,a+a^3).$ Here $3^2\equiv 4 \pmod 5.$ The matrix $A^2$ is a $4$-circulant matrix. Then $A^2[0,0]=a^6+1=A^2[1,4].$ For $A$ to be involutory, we must have $A^2[0,0]=1 $ and $A^2[1,4]=0$, which is not possible. Consequently, it is evident that $A$ is never involutory.
\end{example}

The above theorem implies that, to construct a $g$-circulant matrix with MDS and involutory properties, we only need to focus on the case $g^2 \equiv 1 \pmod {k}.$
We begin with few elementary lemmas and a theorem to determine the number of solutions of the equivalence relation $x^2 \equiv 1 \pmod k.$ For the sake of completeness here we record some elementary proofs.

\begin{lemma}\label{2^m case}
Let $k=2^m$, $m$ positive integer. Then, the number of solutions to the congruence relation $x^2 \equiv 1 \pmod k$ in the residue modulo $k$ is $$\begin{cases}
  1, & \text{if $m=1$ };\\
  2, & \text{if $m=2$ };\\
  4, & \text{if $m \geq 3$}.
\end{cases}$$
\end{lemma}

\begin{proof}
For $k=2$ the only solution of the congruence relation is $1.$ When $k=4$ there are two solutions, namely $x=1,3.$
 
Consider the case $k=2^m, m \geq 3.$ It is apparent that $\pm 1 \pmod {2^m}$ are the trivial solutions. Let there exists a non-trivial solution $\alpha$ such that $\alpha^2\equiv 1 \pmod {2^m}$. This implies $(\alpha+1)(\alpha-1) \equiv 0 \pmod {2^m}$. Furthermore, $\alpha$ must be an odd number. Let $\alpha=2k+1$ for some integer $k \geq 0.$ This implies $4k(k+1)\equiv 0 \pmod {2^m}.$ Therefore $2^{m-2}|k(k+1).$ If $k$ is even then it must be divisible by $2^{m-2}$ and hence $\alpha=2^{m-1}l+1.$ If $k+1$ is even then $\alpha=2^{m-1}l'-1.$ Therefore, the other two solutions correspond to the choices $l=1$ and $l'=1$. Hence, the complete set of solutions modulo $k$ is $\pm 1$ and $(2^{m-1}\pm 1)$.

% We will now prove that $2^{m-1}\pm 1 \pmod {2^m}$ also solutions. Consider the square of $2^{m-1} \pm 1.$ Then $(2^{m-1}\pm 1)^2=2^{2(m-1)}+1 \pm 2^m= 1 \pmod {2^m}$, since $2^{2(m-1)}\pm 2^m \equiv 0 \pmod {2^m}.$ 
%
%To establish that these are the only solutions, we will assume that there exists an $\alpha$  such that $\alpha^2 \equiv 1 \pmod {2^m}$ and $\alpha \neq \{\pm 1, 2^{m-1}\pm1\}.$ 
%Furthermore, $\alpha$ must be an odd number and assume $\alpha < 2^{m-1}.$ Then $\alpha$ can be expressed as $\alpha=2^{m-i}\pm 1$ for some $i$ in range $ 2 \leq i \leq m-1.$ Then $\alpha^2=2^{2(m-i)}+1\pm 2^{m-i+1}=2^{m-i+1}(2^{m-i+1} \pm 1)+1 < 2^{m-1}(2^{m-1} \pm1)+1$ and therefore $\alpha^2$ not congruent to $1$ modulo $2^m.$
\end{proof}

Next, consider the case for an odd prime power in the following lemma.
\begin{lemma}\label{p^m case}
Let $k = p^m$, where $p \geq 3$ be a prime number. Then the solutions to the congruence $x^2 \equiv 1 \pmod{k}$ in the residue modulo $k$ are given by $x \equiv \pm 1 \pmod{k}.$
\end{lemma}

\begin{proof}
Let $x^2 \equiv 1 \pmod {p^m}.$ This implies $p^m |(x+1)(x-1).$ Suppose that $x \neq \pm 1 \pmod {p^m}.$ In this case, both $x+1$ and  $ x-1$ are less than $p^m.$ This implies $p$ divides both $x+1$ and $x-1.$ Therefore $p|2$, which is a contradiction to $p$ is an odd prime.
\end{proof}

Applying the Chinese Remainder Theorem one can prove the following theorem.
\begin{theorem}\label{general n case}
Let $k=2^mp_1^{m_1}\cdots p_l^{m_l}$ where $p_i$'s are odd primes and $m,m_i \geq 0$ for $1 \leq i \leq l.$ Then the number of solutions of the equation $x^2=1 \mod k$ in residue modulo $k$ is $$\begin{cases}
  2^l, & \text{if $m=0,1$ };\\
  2^{l+1}, & \text{if $m=2$ };\\
  2^{l+2}, & \text{if $m \geq 3$}.
\end{cases}$$
\end{theorem}

\begin{proof}
Using Chinese Remainder Theorem, Lemma \ref{2^m case} and Lemma \ref{p^m case} the number of solutions of $x^2 \equiv 1 \pmod k$ is $2^i\cdot 2^l$ where  $$i=\begin{cases}
  0, & \text{if $m=0,1$ };\\
  1, & \text{if $m=2$ };\\
  2, & \text{if $m \geq 3$}.
\end{cases}$$
\end{proof}

In \cite{LS}, the authors conjectured that involutory MDS cyclic matrices of order $4,8$  over the finite field $\mathbb{F}_{2^m}$ do not exists. In the following theorem, we substantiate their conjecture within a specific subclass of cyclic matrices. 
Specifically, we show that $g$-circulant involutory matrices of order $2^d \times 2^d$ cannot be MDS by proving the existence of a singular submatrix. To establish this, we first demonstrate the presence of a left-circulant matrix as a submatrix in a $g$-circulant matrix of order $2^d \times 2^d$ for a particular $g.$
\begin{lemma}\label{left-circ submatrix}
Let $A$ be a $g$-circulant matrix of order $2^d \times 2^d$ and $g=2^{d-1}-1.$ Let $(c_0,c_1,c_2,\hdots, c_{2^d-1})$ be the first row of $A.$ Then $A$ has two left-circulant submatrices of order $2^{d-1} \times 2^{d-1}.$
\end{lemma}
\begin{proof}
Let $(c_0,c_1,c_2,\hdots, c_{2^d-1})$ be the first row of $A.$ We denote the $i$-th row of $A$ by $R_i$, $j$-th column by $C_j$ and $A[0,j]=c_j$ for $0 \leq j \leq 2^d-1.$ Consider the entries of the row $R_2$: 
\begin{eqnarray*}
A[2,0]=c_{2^d-2g \pmod {2^d}}=c_{2^d-2(2^{d-1}-1)}=c_2,  A[2,1]=c_3, \\A[2,2]=c_4, \hdots, A[2,2^d-1]=c_{2+2^d-1 \pmod{2^d}}=c_1.
\end{eqnarray*}
Similarly, the entries of the row $R_4$ are:
\begin{eqnarray*}
 A[4,0]=c_{2^d-4g \pmod {2^d}}=c_4, A[4,1]=c_5,\\ A[4,2]=c_6, \hdots, A[4,2^d-1]=c_{4+2^d-1 \pmod{2^d}}=c_3.
\end{eqnarray*}
Continuing this process, we find the entries of row $R_{2^d-2}$ are :
\begin{eqnarray*}
A[2^d-2,0]=c_{2^d-2},A[2^d-2,1]=c_{2^d-1}, \hdots, A[2^d-2,2^d-1]=c_{2^d-3}. 
\end{eqnarray*}
Therefore the rows $R_0,R_2,R_4,\hdots,R_{2^d-1}$ and the columns $C_0,C_2,C_4,\hdots,C_{2^d-2}$ form the left-circulant matrix with the entries of first row $c_0,c_2,c_4,\hdots,c_{2^d-2}.$ Similarly, the rows $R_0,R_2,R_4,\hdots,R_{2^d-1}$ and the columns $C_1,C_3,C_5,\hdots,C_{2^d-1}$ form the left-circulant matrix with the entries of the first row $c_1,c_3,c_5,\hdots,c_{2^d-1}.$
\end{proof}
Note that, if $\mathcal{C}_1$ is a circulant matrix with first row $(c_0,c_1,c_2,\hdots, c_{k-1})$ and $\mathcal{C}_2$ is a left-circulant matrix with the same first row, then their determinant is same over the finite field of characteristic $2.$
We are now ready to prove the theorem.

\begin{theorem}\label{2^d g cir non mds}
Let $A$ be a $g$-circulant matrix of order $2^d \times 2^d$ over a finite field of characteristic $2$ and $\gcd(g,2^d)=1.$ Let $(c_0,c_1,c_2,\hdots, c_{2^d-1})$ be the first row of $A$ and $g^2 \equiv 1 \pmod {2^d}.$ If $A$ is an involutory matrix, then $A$ cannot be an MDS matrix.
\end{theorem}
\begin{proof}
Consider the $g$-circulant matrix $A$ of order $2^d \times 2^d$  over $\mathbb{F}_{2^m}$ with first row $(c_0,c_1,c_2,\hdots, c_{2^d-1}).$ Given that $g^2 \equiv 1 \pmod {2^d}$, we consider the following cases for the possible values of $g$:

\textbf{Case I.} For the case $g=1,~A$ is a circulant matrix. If $A$ is involutory, then from  Lemma $9$ of \cite{GR}, $A$ can not be MDS. 

%\textbf{Case II.} For the case $g= 2^d-1$ the matrix $A$ becomes a left-circulant matrix. Let $A$ be involutory. Since left-circulant matrices are symmetric, which implies they are orthogonal. Therefore $A$ is a $2^d\times 2^d$ left-circulant, orthogonal matrix. Therefore $A$ cannot be MDS followed by by the Theorem $5.4$ of \cite{TA2}.

Let $g=2^d-1$. Then $A$ becomes a left-circulant matrix. Suppose that $A$ is involutory. Since every left-circulant matrix is symmetric, it follows that $A$ is orthogonal. Hence $A$ is a $2^d \times 2^d$ $g$-circulant orthogonal matrix over the finite field $\mathbb{F}_{2^m}$ with $\gcd(g,2^d)=1$.

By Theorem $5.4$ of \cite{TA2}, a $2^d \times 2^d$ $g$-circulant orthogonal matrix over $\mathbb{F}_{2^m}$ with $\gcd(g,2^d)=1$ cannot be an MDS matrix. Therefore, $A$ is not an MDS matrix.

\textbf{Case III.}  From the Lemma \ref{2^m case}, there exist values of $g$ in the range $1<g<2^d-1$ satisfying $g^2 \equiv 1 \pmod {2^d}.$ Note that $g$ must be an odd number since $\gcd(2^d,g)=1.$
 
Moreover, $A^2$ is a circulant matrix according to Lemma \ref{gh-circulant}. Let $A$ be involutory. Since $Q_g^2$ is a $g^2$-circulant matrix and $g^2 \equiv 1 \pmod {2^d}$, we have $Q_g^2=I.$ This implies $Q_g^2P^l=P^l$ for $0 \leq l \leq k-1.$ By utilizing Theorem \ref{A^2 g-circulant structure} and the involutory property, we can deduce that $A^2[0,0]=1$  and $A^2[0,l]=0$ for $1 \leq l \leq 2^d-1.$
We calculate the coefficient of  $A^2[0,2^{d-1}].$
\begin{align*}
A^2[0,2^{d-1}] &=\sum_{\substack{i,j=0,\\gi+j=2^{d-1} \pmod {2^d}}}^{k-1}c_ic_j \\&= \sum_{\substack{i=0,\\gi+i=2^{d-1} \pmod {2^d}}}^{k-1}c_i^2+\sum_{\substack{i\neq j,~i,j=0,\\gi+j=2^{d-1} \pmod {2^d}}}^{k-1}c_ic_j 
\end{align*} 
Consider the equation $gi+j=2^{d-1} \pmod {2^d}.$ Since $g$ is invertible, multiply this equation by $g^{-1}$  and using that $g$ has self-inverse, we get $i+gj=g^{-1}2^{d-1}=g2^{d-1}=(2k_1+1)2^{d-1}=2^dk_1+2^{d-1}=2^{d-1} \pmod {2^d}.$ Therefore the set $\{(i,j):gi+j=2^{d-1} \pmod {2^d}\}$ is same as the set $\{(i,j):i+gj=2^{d-1} \pmod {2^d}\}.$ %holds because $g2^{d-1}=(2k_1+1)2^{d-1}=2^dk_1+2^{d-1}=2^{d-1} \pmod {2^d}$. 
Thus, the equation reduces as follows:
\begin{align*}
A^2[0,2^{d-1}] &= \sum_{\substack{i=0,\\gi+i=2^{d-1} \pmod {2^d}}}^{k-1}c_i^2+ \sum_{\substack{i<j,~i,j=0,\\gi+j=2^{d-1} \pmod {2^d}}}^{k-1}2c_ic_j\\&=\sum_{\substack{i=0,\\(g+1)i=2^{d-1} \pmod {2^d}}}^{k-1}c_i^2
\end{align*}

Consider the set $S=\{i:(g+1)i=2^{d-1} \pmod {2^d}\}.$ Note that, if $\alpha \in S$, then the additive inverse of $\alpha$ also belongs to $S$ because  $(g+1)(2^d-\alpha)=-(g+1)\alpha=2^d-2^{d-1}=2^{d-1}  \pmod {2^d}.$ As a result, $|S|$ is even.

From Lemma \ref{2^m case}, the only possibilities for $g$ are $2^{d-1}\pm 1.$ First we prove that for $g=2^{d-1}-1$, if $\alpha \in S$ then $\alpha+2 \in S.$ This is evident because $(g+1)(\alpha+2)=2^{d-1}+2(g+1)=2^{d-1}+2^{d}=2^{d-1} \pmod{2^d}.$ Since $1 \in S$ in this scenario, we get $1, 1+2=3,5,\hdots, 2^d-1 \in S.$ Also, $2 \notin S.$ Hence,
\begin{eqnarray*}
A^2[0,2^{d-1}]=(c_1+c_3+\cdots+c_{2^d-1})^2
\end{eqnarray*}
Since $A$ is involutory this implies $c_1+c_3+\cdots+c_{2^d-1}=0.$ 

Therefore by Lemma \ref{left-circ submatrix},  we get a left-circulant submatrix of order $2^{d-1} \times 2^{d-1}$ with determinant $0$ and this implies $A$ is not an MDS matrix.

Next consider the case for $g=2^{d-1}+1.$ First, we prove that if $\alpha \in S$ then $\alpha+2^{d-1}\in S.$ This hold because $(g+1)(\alpha+2^{d-1})=2^{d-1}+2^{d-1}(g+1)=2^{d-1}+2^{d-1}(2^{d-1}+2)=2^{d-1}+2^{d}(2^{d-2}+1)=2^{d-1} \pmod{2^d}.$ Furthermore, $2^{d-2} \in S$ because $(g+1)2^{d-2}=(2^{d-1}+2)2^{d-2}=2^{d-1}(2^{d-1}+1)=2^{d-1} \pmod{2^d}.$ Consequently $2^{d-2},2^{d-2}+2^{d-1} \in S .$ Thus 
\begin{eqnarray*}
A^2[0,2^{d-1}]=(c_{2^{d-2}}+c_{3\cdot 2^{d-2}})^2
\end{eqnarray*}
 Since $A$ is involutory, this implies $(c_{2^{d-2}}+c_{3 \cdot 2^{d-2}})=0.$ Consider the $2 \times 2 $ submatrix of $A$ with entries $A[0,2^{d-2}],A[0,3\cdot 2^{d-2}],A[2^{d-1},2^{d-2}]$ and $A[2^{d-1},3\cdot 2^{d-2}].$ The entries in the first rows are $A[0,2^{d-2}]=c_{2^{d-2}}$ and $A[0,3\cdot 2^{d-2}]=c_{3\cdot 2^{d-2}}.$ By calculating the entries of $2^{d-1}$-th row, we get  
 \begin{align*}
A[2^{d-1},2^{d-2}]&=c_{2^d-2^{d-1}g+2^{d-2} \pmod {2^d}}\\&=c_{2^d-2^{d-1}(2^{d-1}+1)+2^{d-2} \pmod{2^d}}=c_{3\cdot 2^{d-2}},\\
A[2^{d-1},3\cdot2^{d-2}]&=c_{2^d-2^{d-1}g+3 \cdot 2^{d-2} \pmod {2^d}}\\&=c_{2^{d-2}(2^2-2^d-2+3)\pmod{2^d}}=c_{2^{d-2}}
\end{align*}
 
Hence there exists a $2 \times 2$ submatrix of $A$ with determinant $0.$ Thus $A$ is not MDS. 
\end{proof}

Next we consider $g$-circulant matrices of orders other than $2^d \times 2^d.$ Let $k=2^m\prod_{i=1}^l p_i^{m_i}, m \geq 0, m_i\geq 1$ and $p_i$'s are odd primes. Note that this representation of $k$ covers all natural numbers other than powers of $2.$

\begin{theorem}\label{other odd order}
Let $A$ be a $g$-circulant matrix of order $k \times k$ with $\gcd(g,k)=1$ over a finite field of characteristic $2$ with $k=2^m\prod_{i=1}^l  p_i^{m_i}, m \geq 0, m_i \geq 1$ and $p_i$'s are odd primes. Let $(c_0,c_1,c_2,\hdots, c_{k-1})$ be the first row of $A$ and $~g^2 \equiv 1 \pmod k.$ If $A$ is an involutory matrix and $1 \leq g < k-1$, then $A$ is not an MDS matrix.
\end{theorem}

\begin{proof}
\textbf{Case I.}
Let $g=1$ i.e., $A$ is a circulant matrix. Then from Lemma $9$ of \cite{GR}, $A$ is not MDS.

\textbf{Case II.}
 Let's consider the case $1 < g< k-1.$  
 %$k=2^mp_1^{m_1}p_2^{m_2}\cdots p_l^{m_l}$. 
 According to Theorem \ref{general n case}, there exists $g$ in this range with $g^2 \equiv 1 \pmod k.$ %$\gcd(g,k)=1$. 
 Therefore $A^2$ is a circulant matrix by Lemma \ref{gh-circulant}. We now calculate the entry $A^2[0,g+1]$ using Theorem \ref{A^2 g-circulant structure} :
 \begin{align*}
A^2[0,g+1] &=\sum_{\substack{i,j=0,\\gi+j=g+1 \pmod {k}}}^{k-1}c_ic_j \\&= \sum_{\substack{i=0,\\gi+i=g+1 \pmod {k}}}^{k-1}c_i^2+\sum_{\substack{i\neq j,~i,j=0,\\gi+j=g+1 \pmod {k}}}^{k-1}c_ic_j
\end{align*} 
Let $(i,j)$ satisfy the equation $gi+j=g+1 \pmod k.$ Then $(j,i)$ also satisfies the same because $g$ has self-inverse. This implies $i+gj=1+g^{-1}=1+g \pmod k.$ Therefore we can write $A^2[0,g+1]$ as the following:
\begin{align*}
A^2[0,g+1] &= \sum_{\substack{i=0,\\gi+i=g+1 \pmod {k}}}^{k-1}c_i^2+ \sum_{\substack{i< j,~i,j=0,\\gi+j=g+1 \pmod {k}}}^{k-1}2c_ic_j\\&=\sum_{\substack{i=0,\\(g+1)i=g+1 \pmod {k}}}^{k-1}c_i^2
\end{align*} 
Consider the set $S=\{i:gi+i=g+1 \pmod k\}.$ This set is non-empty because $1 \in S.$ Note that, there always exists a smallest non-zero integer $\alpha<k$  such that $(1+g)\alpha=0 \pmod k.$ This holds because, the conditions $k |(g+1)(g-1)$ and $1 \leq (g+1),(g-1)< k$ implies $\gcd(k,g+1) >1.$ Therefore, such $\alpha$ exists.

Then $1+\beta \alpha \in S$ for $ \beta=\{1,2,\hdots,\lfloor \frac{k-1}{\alpha} \rfloor\},$ because $(1+g)(1+\beta \alpha)=(1+g)+\beta \alpha(1+g)= 1+g \pmod k.$ Therefore $A^2[0,g+1]$ can be written as:
\begin{align*}
A^2[0,g+1] &=\sum_{\substack{i=0,\\(g+1)i=g+1 \pmod {k}}}^{k-1}c_i^2 \\&=c_1^2+c_{1+\alpha}^2+c_{1+2\alpha}^2+\cdots+c_{1+\lfloor \frac{k-1}{\alpha} \rfloor \alpha}^2\\ &=(c_1+c_{1+\alpha}+c_{1+2\alpha}+\cdots+c_{1+\lfloor \frac{k-1}{\alpha} \rfloor \alpha})^2
\end{align*} 
Since $A$ is involutory, we get $c_1+c_{1+\alpha}+c_{1+2\alpha}+\cdots+c_{1+\lfloor \frac{k-1}{\alpha} \rfloor \alpha}=0.$
Therefore the determinant of the submatrix of $A$ with entries in rows $R_0,R_{1+\alpha},R_{1+2\alpha},\hdots, R_{1+\lfloor \frac{k-1}{\alpha} \rfloor \alpha}$ and columns $C_1,C_{1+\alpha},C_{1+2\alpha},\hdots, C_{1+\lfloor \frac{k-1}{\alpha} \rfloor \alpha}$ is $0.$ Therefore $A$ is not MDS.
\end{proof}

In the next case, we prove that it is possible to construct left-circulant involutory MDS matrices of order other than $2^d$ under certain conditions. This result resembles Proposition $6$ in \cite{LS}, but our proof utilizes Theorem \ref{A^2 g-circulant structure}, which provides a clear method for expressing the entries of $A^2.$

\begin{theorem}\label{left-circulant MDS inv}
Let $A$ be a left-circulant matrix of order $k \times k$ over a finite field of characteristic $2$ with $k=2^m\prod_{i=1}^l  p_i^{m_i}, m \geq 0, m_i\geq 1$ and $p_i$'s are primes. Let $(c_0,c_1,c_2,\hdots, c_{k-1})$ be the first row of $A.$ Then $A$ is involutory if and only if the following conditions hold:
\begin{enumerate}
\item $\sum\limits_{i=0}^{k-1} c_i=1$,
\item $\sum\limits_{i=0}^{k-1}c_ic_{i+l}=0, ~1 \leq l \leq \lfloor \frac{k-1}{2} \rfloor$.
\end{enumerate}
\end{theorem}

\begin{proof}
Consider a left-circulant matrix $A$ of order $k \times k$ over the finite field of characteristic $2$ with $k=2^m\prod_{i=1}^l  p_i^{m_i}$ and $p_i$'s are primes. Since $g\equiv-1 \pmod k$, then from Lemma \ref{gh-circulant}, $A^2$ is circulant. Let $A$ be involutory. Then $A^2[0,0]=1$ and $A^2[0,l]=0, 1 \leq l \leq k-1.$ Therefore using Theorem \ref{A^2 g-circulant structure}, $A^2[0,0]$ can be written as:
\begin{align*}
A^2[0,0]=\sum_{\substack{i,j=0,\\gi+j=0 \pmod k}}^{k-1}c_ic_j = \sum_{\substack{i=0,\\gi+i=0 \pmod k}}^{k-1}c_i^2 = (\sum_{i=0}^{k-1}c_i)^2.
\end{align*}
This holds because $gi+j=0 \pmod k$ and  $g\equiv-1 \pmod k$ implies $j=i.$ Thus $ \sum_{i=0}^{k-1} c_i=1.$

%for all $i=0,1,\cdots,k-1$, we have $(g+1)i=0 \pmod k$. 

The coefficient of $P^l$ is $\sum\limits_{\substack{i,j=0,\\gi+j=l \pmod k, l \neq 0}}^{k-1}c_ic_j$ for $1 \leq l \leq k-1.$ Since $g \equiv-1 \pmod k$, the coefficient can be written as  $\sum\limits_{i=0}^{k-1}c_ic_{i+l}$ for $1 \leq l \leq k-1,$ where the indices are taken modulo $k$. 
%Therefore
%\begin{align*}
%A^2[0,l]=\sum_{\substack{i,j=0,\\j-i=l \pmod k, l \neq 0}}^{k-1}c_ic_j
%\end{align*}
 Since $g^2 \equiv 1 \pmod k$, the equation $gi+j=l\pmod k$ can be written as $i+gj=gl=k-l\pmod k.$ Hence the set $\{(i,j): gi+j=l \pmod k\}$ same as the set $\{(i,j): i+gj=k-l\pmod k\}.$
 
Therefore, it is enough to consider first $\lfloor \frac{k-1}{2} \rfloor$ entries of first row of $A^2$, i.e., coefficients of $P^l$ with $ 1 \leq l \leq \lfloor \frac{k-1}{2} \rfloor.$ 
 Note that, when $k$ is even, 
\begin{align*}
A^2\left[0,\frac{k}{2}\right]=\sum_{\substack{i,j=0,\\gi+j=\frac{k}{2} \pmod k}}^{k-1}c_ic_j =\sum_{\substack{i,j=0,\\gi+j=\frac{k}{2} \pmod k, i<j}}^{k-1}2c_ic_j=0 .
\end{align*} 
This holds  because the set $\{(i,j): gi+j=\frac{k}{2} \pmod k\}$ equals to  $\{(i,j):i+gj=g^{-1}\frac{k}{2}=g\frac{k}{2}=\frac{k}{2} \pmod k\}.$ This implies for even $k$, $A^2[0,\frac{k}{2}]$ always $0.$ Therefore the conditions hold. 

Conversely, if the conditions hold, then $A^2[0,0]=1$ and $A^2[0,l]=0$ for $1 \leq l \leq k-1.$ Since $A^2$ is circulant, this implies $A$ is involutory. Hence proved. 
 \end{proof}

\begin{remark}
Theorem \ref{left-circulant MDS inv} provides an efficient criterion for identifying involutory left-circulant matrices.
Instead of computing $A^2$ directly, this result yields explicit expressions for the entries of
$A^2$ in terms of the first row of $A$. Over fields of characteristic $2$, these expressions
simplify, and the involutory property reduces to checking
$1+\lfloor (k-1)/2 \rfloor$ algebraic equations in the coefficients $c_0, c_1, \cdots, c_{k-1}$.
This is substantially simpler than performing full matrix multiplication. Once the involutory property is satisfied, the MDS property can be verified separately by checking the determinants of all square submatrices.
\end{remark}

Consider the following example of left-circulant involutory MDS matrix from \cite{LS}.
\begin{example}
Let $a$ be a primitive element of the finite field $\mathbb{F}_{2^8}$ with the generating polynomial $1+x^2+x^5+x^6+x^8.$ Construct the left-circulant matrix of order $5 \times 5$ with the first row $(1,a,1+a+a^4+a^5+a^7,1+a+a^3+a^4+a^5+a^7,a+a^3).$ Then 
$$\qquad\sum\limits_{i=0}^{k-1} c_i=1+a+1+a+a^4+a^5+a^7+1+a+a^3+a^4+a^5+a^7+a+a^3=1,$$
\begin{align*}
\sum\limits_{\substack{i,j=0,\\4i+j=1 \pmod k}}^{k-1}c_ic_j&=c_0c_1+c_1c_2+c_2c_3+c_3c_4+c_4c_0\\
&=(1\cdot a)+(a\cdot (1+a+a^4+a^5+a^7))+((1+a+a^4+a^5+a^7)\\&\cdot(1+a+a^3+a^4+a^5+a^7))+((1+a+a^3+a^4+a^5+a^7)\\&\cdot(a^3+a))+((a^3+a))\\
&=0,\\
\sum\limits_{\substack{i,j=0,\\4i+j=2 \pmod k}}^{k-1}c_ic_j&=c_0c_2+c_1c_3+c_2c_4+c_3c_0+c_4c_1\\
&=(1\cdot (1+a+a^4+a^5+a^7))+(a\cdot (1+a+a^3+a^4+a^5+a^7))+\\&((1+a+a^4+a^5+a^7)\cdot(a^3+a))+((1+a+a^3+a^4+a^5+a^7))\\&+((a^3+a)\cdot a)\\
&=0.
\end{align*}
Therefore, $A^2=I$ from Theorem \ref{left-circulant MDS inv}. Additionally, all submatrices of $A$ are non-singular, affirming that $A$ is an MDS matrix as well.
\end{example} 

\begin{table*}
\centering
\begin{tabular}{||c c c c ||} 
 \hline
 Type of MDS Matrix & Order & Existence & Reference \\ [0.5ex] 
 \hline\hline
 circulant, involutory  & $n \times n$ & do not exist &  \cite{GR}  \\ 
\hline    
 left- circulant, involutory &  $2^d \times 2^d$ & do not exist & \cite{LS}  \\
 \hline
 left- circulant, involutory & other orders & may exists & \cite{LS}\\
 \hline
 \makecell{g-circulant, involutory,\\
 g odd, $g^2 \equiv 1 \pmod {2^d}$}  
 &  $2^d \times 2^d$ & do not exist & Theorem \ref{2^d g cir non mds} \\
 \hline
 \makecell{g-circulant, involutory,\\
 $\gcd(g,k)=1$, $g^2 \equiv 1 \pmod k$,\\ $g< k-1$}
 & \makecell{$k=2^m\prod_{i=1}^l  p_i^{m_i}$\\$ \ m \geq 0,\ m_i \geq 1$} 
 & do not exist & Theorem \ref{other odd order} \\  
 \hline   
 cyclic, semi-orthogonal &  $2^d \times 2^d$ & do not exist & Theorem \ref{cyclic si} \\
 \hline 
 \end{tabular}\vspace{5pt}

\caption{List of pre-existing results and new results over finite fields of characteristic $2$.
 \label{tab:table 1}}
\end{table*}

Till now, we have discussed 
$
g$-circulant matrices with involutory properties, and the summary is given in Table~\ref{tab:table 1}. In the next section, we focus on 
$
g$-circulant matrices with semi-involutory and semi-orthogonal properties.
\section{\bf g-Circulant matrices with semi-involutory and semi-orthogonal properties} \label{sec:g-circulant MDS semi-involutory}
In this section, we focus on $g$-circulant matrices endowed with semi-involutory and semi-orthogonal properties. In $2022,$ Chatterjee {\it{et al.}} \cite{TA} demonstrated that, for a circulant matrix of order $k \times k$, possessing the semi-orthogonal property leads to the intriguing result that the $k$-th power of the associated diagonal matrices yield a scalar matrix. This finding prompts a natural question: does this distinctive characteristic also hold true for $g$-circulant semi-orthogonal matrices? Our investigation extends to this inquiry, considering the case where $\gcd(g,k)=1$, as this condition proves to be essential for the non-singularity of the matrix in \cite{TA2}.

\begin{theorem}\label{gcircsemiortho 1}
Let $A$ be a $g$-circulant matrix of order $k \times k$ over a finite field $\mathbb{F}^*$ with $\gcd(g,k)=1.$ Then $A$ is semi-orthogonal if and only if there exist non-singular diagonal matrices $D_1, D_2$ such that $D_1^k=k_1I$ and $D_2^k=k_2I$ for non-zero scalars $k_1, k_2$ in the finite field and $A^{-T}=D_1AD_2.$
\end{theorem}

\begin{proof}
Let $A$ be a $g$-circulant matrix with semi-orthogonal property. Then there exists non-singular diagonal matrices $D_1$ and $D_2$ such that $A^{-T}=D_1AD_2.$ Since $A$ is $g$-circulant and $\gcd(g,k)=1$, by Theorem \ref{relation cyclic circulant}, there exists a unique permutation matrix $Q$ such that $AQ=C$, where $C$ is a circulant matrix. Lemma \ref{per equi s.o} implies that $C$ is also semi-orthogonal. According to Theorem \ref{circ semi-ortho}, the associated diagonal matrices $E_1, E_2$ of $C$ satisfy $E_1^k=k_1I$ and $E_2^k=k_2I$ for some non-zero scalars $k_1, k_2$ in the finite field and $C^{-T}=E_1CE_2.$ This implies $(AQ)^{-T}=E_1AQE_2.$ Thus $A^{-T}=E_1AQE_2Q^T.$ Consider $D_1=E_1$, which implies $D_1^k=E_1^k=k_1I.$ Let $D_2=QE_2Q^T.$ If $E_2=$ diagonal$(e_0,e_1,\hdots,e_{n-1})$ and $\sigma \in S_k$ be the permutation associated to $Q$,  
then $D_2$ is also a diagonal matrix with diagonal entries $(e_{\sigma^{-1}(0)},e_{\sigma^{-1}(1)},e_{\sigma^{-1}(2)},\hdots,e_{\sigma^{-1}(k-1)}).$ 
Since $e_{\sigma^{-1}(i)}=e_j$ for $0 \leq i,j \leq k-1$ and $QQ^T=I$, then $D_2^k=(QE_2Q^T)^k=k_2I.$ Hence proved.

Conversely, if there exists non-singular diagonal matrices $D_1,D_2$ such that $D_1^k=k_1I$ and $D_2^k=k_2I$ for non-zero scalars  $k_1,k_2$ in the finite field and $A^{-T}=D_1AD_2$, then by the definition $A$ is semi-orthogonal.
\end{proof}

Recently, Kumar {\it{et al.}}  provided an explicit characterization of the diagonal matrices associated with circulant semi-orthogonal matrices in \cite{KSMG}. The result is following:

\begin{theorem}[{\cite[Theorem 12]{KSMG}}]\label{circ s-o}
Let $A$ be an $n\times n$ semi-orthogonal circulant matrix over the finite field $\mathbb{F}_{p^m}$. Then its associated diagonal matrices $D_1$ and $D_2$ are given by
$
D_1=\alpha_1\operatorname{Diag}(1,\zeta,\zeta^2,\ldots,\zeta^{n-1})
$
and
$
D_2=\beta_1\operatorname{Diag}(1,\zeta^{n-1},\zeta^{n-2},\ldots,\zeta),
$
where $\zeta$ is an $n$-th root of unity in $\mathbb{F}_{p^m}^{\ast}$ and
$\alpha_1,\beta_1\in\mathbb{F}_{p^m}^{\ast}$.
\end{theorem}
Combining this result with Theorem \ref{gcircsemiortho 1}, we obtain an explicit description of the associated diagonal matrices of $g$-circulant semi-orthogonal matrices.

\begin{remark}
Let $A$ be a $g$-circulant semi-orthogonal matrix of order $k\times k$. By Theorem~\ref{gcircsemiortho 1}, the associated diagonal matrices $D_1$ and $D_2$ of $A$ satisfy
$
D_1=E_1
\quad\text{and}\quad
D_2=QE_2Q^T,
$
where $Q$ is a permutation matrix and $E_1,E_2$ are the associated diagonal matrices of the corresponding circulant semi-orthogonal matrix.

Using Theorem~\ref{circ s-o}, it follows that
$
D_1=\alpha_1\operatorname{Diag}(1,\zeta,\zeta^2,\ldots,\zeta^{k-1}),
$
and
$
D_2
=
Q\Bigl(
\beta_1\operatorname{Diag}(1,\zeta^{k-1},\zeta^{k-2},\ldots,\zeta)
\Bigr)Q^T,
$
where $\zeta$ is a $k$-th root of unity in $\mathbb{F}_{p^m}^{\ast}$ and
$\alpha_1,\beta_1\in\mathbb{F}_{p^m}^{\ast}$.
If $\sigma\in S_k$ denotes the permutation corresponding to $Q$, then
$
D_2=
\beta_1\operatorname{Diag}
\bigl(
\zeta^{k-\sigma(0)},
\zeta^{k-\sigma(1)},
\ldots,
\zeta^{k-\sigma(k-1)}
\bigr).
$
Thus, the associated diagonal matrices of a $g$-circulant semi-orthogonal matrix admit a characterization analogous to that of circulant semi-orthogonal matrices, differing only by a permutation of their diagonal entries.
\end{remark}

In \cite{KSMG}, the authors established the non-existence of circulant semi-orthogonal MDS matrices of order a power of $2$ over finite fields of characteristic $2$. Their result is stated below.

\begin{theorem}[{\cite[Theorem 13]{KSMG}}]\label{nonexists1}
Semi-orthogonal circulant matrices of order $2^n$ over $\mathbb{F}_{2^m}$ are non-MDS.
\end{theorem}

Using Theorem~\ref{relation cyclic circulant}, we obtain the following immediate generalization to the class of $g$-circulant matrices.
\begin{theorem}\label{gcirc semi-ortho non mds}
Let $A$ be a $2^d \times 2^d$ $g$-circulant semi-orthogonal matrix  over $F_{2^m}$. Then $A$ is not an  MDS matrix.
\end{theorem}
\begin{proof}
Let $A$ be a $2^d \times 2^d$ $g$-circulant semi-orthogonal matrix over $F_{2^m}$. If $\gcd(g, 2^d)>1$, then by Theorem $4.1$ of \cite{TA2}, $A$ cannot be MDS.

Let $\gcd(g, 2^d)=1$. Since $A$ is a $g$-circulant semi-orthogonal matrix, Theorem~\ref{relation cyclic circulant} guarantees the existence of a unique permutation matrix $Q$ such that
$
AQ=C,
$
where $C$ is a circulant matrix. Moreover, by Lemma~\ref{per equi s.o}, the matrix $C$ is also semi-orthogonal. Hence, $C$ is a $2^d\times 2^d$ circulant semi-orthogonal matrix over $\mathbb{F}_{2^m}$.

Since permutation equivalence preserves the MDS property, it follows from Corollary \ref{mds PE} that $C$ is also an MDS matrix. This contradicts Theorem \ref{nonexists1}.

Therefore,  $A$ cannot be an MDS matrix.
\end{proof}

Since Theorem~\ref{relation cyclic circulant} holds for the broader class of cyclic matrices, which includes $g$-circulant matrices as a special case, the same argument as in the proof of Theorem~\ref{gcirc semi-ortho non mds} establishes the following result.

\begin{theorem}\label{cyclic si}
Let $A$ be a $2^d \times 2^d$ cyclic semi-orthogonal matrix over $\mathbb{F}_{2^m}$. Then $A$ is not an MDS matrix.
\end{theorem}
%\begin{proof}
%\textcolor{blue}{Let $A$ be a $2^d \times 2^d$ cyclic  semi-orthogonal matrix. Then Theorem \ref{relation cyclic circulant} implies that there exists a unique permutation matrix $Q$ such that $AQ=C$, where $C$ is a circulant matrix. Further, by Lemma \ref{per equi s.o}, the matrix $C$ is also semi-orthogonal. Hence $C$ is a $2^d \times 2^d$ circulant semi-orthogonal matrix over $\mathbb{F}_{2^m}$.}
%
%\textcolor{blue}{Suppose that $A$ is an MDS matrix. Since permutation equivalence preserves the MDS property, Corollary $4$ of \cite{GPRS} implies that $C$ is also an MDS matrix. However, this contradicts Theorem $13$ of \cite{KSMG}, which states that circulant semi-orthogonal matrices of order $2^d \times 2^d$ over fields of characteristic $2$ cannot be MDS. Therefore, $A$ cannot be an MDS matrix.
%}
%\end{proof}
It is worth noting that  semi-orthogonal MDS $g$-circulant matrices of orders other than powers of $2$ exists. The following example illustrates this fact.
 \begin{example}
Consider the $5 \times 5$ matrix $A=$ circulant$(1,1+\alpha+\alpha^3,1+\alpha+\alpha^3,\alpha+\alpha^3,1+\alpha^3+\alpha^4+\alpha^7),$ where $\alpha$ is a primitive element of the finite field $\mathbb{F}_{2^8}$ with the generating polynomial $x^8+x^4+x^3+x^2+1.$  
Note that, $A$ is semi-orthogonal since $A^{-T}=D_1AD_2,$ where $D_1=$ diagonal$(\alpha^2+\alpha, \alpha^7+\alpha^2+1, \alpha^7+\alpha^6+\alpha^5+\alpha^4+\alpha^2,\alpha^5+\alpha^4+\alpha^3+\alpha^2,\alpha^6+\alpha^3+\alpha+1)$ and $D_2=$ diagonal$(\alpha^7+\alpha^6+\alpha^3+\alpha^2+\alpha+1, \alpha^7+\alpha^5+\alpha^3, \alpha^7+\alpha^5+\alpha^4+\alpha^2+1, \alpha^6+\alpha^5+\alpha^2, \alpha^7+\alpha^5+\alpha^4+\alpha^2+\alpha).$ Here $k_1=\alpha^5+\alpha^3+\alpha^2+\alpha$ and $k_2=\alpha^6+\alpha^4+\alpha^3+1.$ $A$ is also an MDS matrix.

\end{example}

We next consider the semi-involutory case. The following theorem provides a generalization of Theorem~\ref{circ semi-inv} to $g$-circulant matrices.

\begin{theorem}\label{gcirc si}
Let $A$ be a $g$-circulant matrix of order $k \times k$ over a finite field $\mathbb{F}^*$ with $\gcd(g,k)=1.$ Then $A$ is semi-involutory if and only if there exist non-singular diagonal matrices $D_1, D_2$ such that $D_1^k=k_1I$ and $D_2^k=k_2I$ for non-zero scalars $k_1, k_2$ in the finite field and $A^{-1}=D_1AD_2.$
\end{theorem}
\begin{proof}
Let $A=g$-circulant$(a_0,a_1,\hdots,a_{k-1})$ and $\rho \in S_k$ be the $k$-cycle associated to $A.$ Assume that $A$ is semi-involutory. This implies the existence of non-singular diagonal matrices $D_1$ and $D_2$ such that $A^{-1}=D_1AD_2.$ Let $D_1=$diagonal$(d_0,d_1,\hdots,d_{k-1})$ and $D_2=$diagonal$(d_0',d_1',\hdots,d_{k-1}'). $ Then the matrix $A^{-1}$ takes the form $$A^{-1}=\begin{bmatrix}
d_0a_0d_0' & d_0a_1d_1' & \cdots & d_0a_{k-1}d_{k-1}' \\
d_1a_{k-g}d_0' & d_1a_{k-g+1}d_1' & \cdots & d_1a_{k-1-g}d_{k-1}' \\
\vdots & \vdots & \cdots & \vdots\\
d_{k-1}a_gd_0' & d_{k-1}a_{g+1}d_1' & \cdots & d_{k-1}a_{g-1}d_{k-1}' 
\end{bmatrix}.$$ Here the suffixes of $a_i$'s are calculated modulo $k.$ Since inverse of a $g$-circulant matrix is $h$-circulant with $gh\equiv 1 \pmod k$, the entries of the second row of $A^{-1}$ are the same as the entries of the first row shifted right by $h$ positions. Since $h< k,$ there exists $l$ such that $0 \leq l \leq k-1$ and $\rho(l)=h.$ Therefore, 
\begin{align*}
d_0a_0d_0'&=d_1a_ld_{h}'\\
d_0a_1d_1'&=d_1a_{l+1}d_{h+1}'\\
& \vdots\\
d_0a_{k-1}d_{k-1}'&=d_1a_{l-1}d_{h-1}'.
\end{align*}
Here all the suffixes of $d_i,~a_i$ and $d_i'$ are calculated modulo $k.$ Note that, the sets $\{l,l+1,\hdots,l-1\}$ and $\{h,h+1,\hdots,h-1\}$ form a complete set of residues modulo $k.$ Then
%This implies that $d_0d_0'=d_1d_h',~ d_0d_1'=d_1d_{h+1}',\hdots,d_0d_{k-1}'=d_1d_{h-1}'$. 
multiplying all these equalities, we get $d_0^k=d_1^k.$ Similarly, entries of the third row are the same as entries of the second row right shifted by $h$ positions, and that implies $d_1a_{\rho^{-1}(i)}d_{i}'=d_2a_{\rho^{-1}(h+i)}d_{h+i}'$ for $i=0,\hdots,k-1$, and the indices are reduced modulo $k$, which leads to $ d_2^k=d_3^k.$ Continuing this process, we get $d_1^k=d_2^k=d_3^k=d_4^k=\cdots=d_{k}^k.$ 
Moreover, in $A^{-1},$ the second column is $g$-shift of the first column by Lemma \ref{transpose g circ}. Therefore $d_0a_0d_0'=d_ha_{k-hg+1}d_1',~d_1a_{k-g}d_0'=d_{h+1}a_{k-(h+1)g+1}d_1',\hdots,d_{k-1}a_gd_0'=d_{h+(k-1)}a_{k-(h-1)g+1}d_1'.$ Multiplying these equations we get $d_0'^k=d_1'^k.$ Applying the same reasoning for the second and the third columns, we get $d_1'^k=d_2'^k.$ Continuing in a similar manner, we conclude that $d_0'^k=d_1^k=d_2'^k=d_3'^k=\cdots=d_{k-1}'^k.$

Conversely, if there exists non-singular diagonal matrices $D_1,D_2$ such that $D_1^n=k_1I$ and $D_2^n=k_2I$ for non-zero scalars  $k_1,k_2$ in the finite field and $A^{-1}=D_1AD_2$, then by the definition $A$ is semi-involutory.
\end{proof}

Kumar {\it et al.} \cite{KSMG} also provided an explicit characterization of the diagonal matrices associated with circulant semi-involutory matrices. Their result is stated below.

\begin{theorem}[{\cite[Theorem 10]{KSMG}}] \label{cic si}
Let $A$ be an $n\times n$ semi-involutory circulant matrix over the finite field $\mathbb{F}_{p^m}$. Then its associated diagonal matrices $D_1$ and $D_2$ are given by
$
D_1=\alpha_1\operatorname{Diag}(1,\zeta,\zeta^2,\ldots,\zeta^{n-1})
$
and
$
D_2=\beta_1\operatorname{Diag}(1,\zeta^{n-1},\zeta^{n-2},\ldots,\zeta),
$
where $\zeta$ is an $n$-th root of unity in $\mathbb{F}_{p^m}^{\ast}$ and
$\alpha_1,\beta_1\in\mathbb{F}_{p^m}^{\ast}$.
\end{theorem}

\begin{remark}
For a $g$-circulant semi-involutory matrix, the associated diagonal matrices need not be generated by powers of a single root of unity. Nevertheless, they can still be described in terms of $k$-th roots of unity.

From the proof of Theorem~\ref{gcirc si} we obtain
$
d_0^k=d_1^k=\cdots=d_{k-1}^k
$
and
$
(d_0')^k=(d_1')^k=\cdots=(d_{k-1}')^k.
$
Consequently,
$
\left(\frac{d_i}{d_0}\right)^k=1
\quad\text{and}\quad
\left(\frac{d_i'}{d_0'}\right)^k=1
$
for every $i=1,2,\ldots,k-1$.

Define
\[
\zeta_i=\frac{d_i}{d_0}
\qquad\text{and}\qquad
\eta_i=\frac{d_i'}{d_0'},
\]
for $i=1,2,\ldots,k-1$. Then
$
\zeta_i^k=\eta_i^k=1,
$
so that each $\zeta_i$ and $\eta_i$ is a $k$-th root of unity in $\mathbb{F}_{p^m}^{\ast}$.

Therefore,
$
D_1
=
d_0\,\operatorname{Diag}
(1,\zeta_1,\zeta_2,\ldots,\zeta_{k-1}),
$
and
$
D_2
=
d_0'\,\operatorname{Diag}
(1,\eta_1,\eta_2,\ldots,\eta_{k-1}),
$
where
$
\zeta_i^k=\eta_i^k=1,
\quad i=1,2,\ldots,k-1.
$

%Thus, although the associated diagonal matrices of a $g$-circulant semi-involutory matrix may not exhibit the cyclic root-of-unity pattern of the circulant case, their diagonal entries are nevertheless completely determined by $k$-th roots of unity in $\mathbb{F}_{p^m}^{\ast}$.
\end{remark}

\begin{example}
Consider the $2 \times 2$ matrix $A=$ circulant$(1,a^2)$, where $a$ is a primitive element of the finite field $\mathbb{F}_{2^2}$ with the generating polynomial $x^2+x+1.$ Note that, $A$ is semi-involutory since $A^{-1}=D_1AD_2,$ where $D_1=$ diagonal$(a,a)$ and $D_2=I_{2 \times 2}.$ Here $k_1=a+1$ and $k_2=1.$ $A$ is also an MDS matrix. 
\end{example}
\begin{example}
Consider the $4 \times 4$ matrix $A=$ circulant$(a,a^3,a^2+a+1,a^3)$, where $a$ is a primitive element of the finite field $\mathbb{F}_{2^4}$ with the generating polynomial $x^4+x+1.$ Note that, $A$ is semi-involutory since $A^{-1}=D_1AD_2,$ where $D_1=$ diagonal$(a^3+1,a^3+1,a^3+1,a^3+1)$ and $D_2=I_{4 \times 4}.$ Here $k_1=a^3+a^2+a$ and $k_2=1.$ 
\end{example}

\section{\bf Conclusion}

This article offers a comprehensive exploration of $g$-circulant involutory matrices. We begin with a rigorous theoretical proof establishing the non-existence of involutory maximum distance separable matrices of order $2^d \times 2^d$ within a specific subclass of cyclic matrices. Additionally, we conduct a thorough analysis of $g$-circulant semi-orthogonal and semi-involutory matrices.
In the context of these matrices, we extend results from the circulant to the $g$-circulant case. Specifically, we demonstrate that the $k$-th power of the associated diagonal matrices of order $k \times k$ yields a scalar matrix, a property previously established for circulant matrices, and we generalize this characteristic to $g$-circulant cases as well.

\textbf{Data Availability}
No datasets were generated or analysed during the current study.

\textbf{Competing interests} The authors declare no competing interests.

\bibliographystyle{plain}

\end{document}